\documentclass[fleqn]{llncs}

\usepackage{latexsym}
\usepackage{amssymb}
\usepackage{amsmath}
\usepackage[all]{xy}

\newcommand{\ignore}[1]{} 

\makeatletter
\newenvironment{prog}{\vspace{0.7ex}\par
\setlength{\parindent}{0.7cm}
\obeylines\@vobeyspaces\tt}{\vspace{0.7ex}\noindent
}
\makeatother
\newcommand{\startprog}{\begin{prog}}
\newcommand{\stopprog}{\end{prog}\noindent}

\newcommand{\depth}{\mathit{depth}}

\newcommand{\id}{{\mathit{id}}} 

\newcommand{\sleq}{\leqslant}

\newcommand{\ri}{{<\!\!\!<}}

\def\defemb#1#2{\expandafter\def\csname #1\endcsname
                              {\relax\ifmmode #2\else\hbox{$#2$}\fi}}
\defemb{cA}{{\cal A}}
\defemb{cB}{{\cal B}}
\defemb{cC}{{\cal C}}
\defemb{cD}{{\cal D}}
\defemb{cE}{{\cal E}}
\defemb{cF}{{\cal F}}
\defemb{cG}{{\cal G}}
\defemb{cH}{{\cal H}}
\defemb{cI}{{\cal I}}
\defemb{cJ}{{\cal J}}
\defemb{cL}{{\cal L}}
\defemb{cM}{{\cal M}}
\defemb{cN}{{\cal N}}
\defemb{cO}{{\cal O}}
\defemb{cP}{{\cal P}}
\defemb{cQ}{{\cal Q}}
\defemb{cR}{{\cal R}}
\defemb{cS}{{\cal S}}
\defemb{cT}{{\cal T}}
\defemb{cU}{{\cal U}}
\defemb{cV}{{\cal V}}
\defemb{cX}{{\cal X}}
\defemb{cZ}{{\cal Z}}

\newcommand{\pos}{{\cP}os}

\newcommand{\var}{{\cV}ar}
\newcommand{\Var}{{\cV}ar}

\newcommand{\dom}{{\cD}om}
\newcommand{\Dom}{{\cD}om}

\newcommand{\Ran}{{\cR}an}

\newcommand{\nil}{[\:]}

\newcommand{\toppos}{\epsilon} 
\newcommand{\ol}[1]{\overline{#1}}  

\newcommand{\ot}{\leftarrow}

\def\ri{<\!\!\!<}   

\def\res{\mathrel{\vert\grave{ }}}

\def \tuple#1{\langle #1 \rangle}

\long\def\comment#1{}

\newcommand{\Hpos}{\mathcal{H}^+}
\newcommand{\Hneg}{\mathcal{H}^-}
\newcommand{\fail}{\mathsf{fail}}
\newcommand{\iclp}{\mathit{\cS\cU}}
\newcommand{\iclpp}{{\cS\cU}^+}
\newcommand{\pplus}{\mathcal{P}^+_\lin(A,\Hpos)}
\newcommand{\mus}{\mathit{max}}

\newcommand{\mgu}{\mathsf{mgu}}
\newcommand{\lin}{\mathsf{lin}}

\newcommand{\sep}{\mathit{\;]\![\;}}

\newcommand{\cand}{{\scriptstyle \:\wedge\:}}

\usepackage{color}

\newcommand{\blue}[1]{{\color{black} #1}}

\begin{document}

\title{On the Completeness of Selective Unification in\\ Concolic Testing of Logic Programs%
\thanks{This work has been partially supported by the EU (FEDER) and the 
Spanish \emph{Ministerio de Econom\'{\i}a y Competitividad} 
under grant TIN2013-44742-C4-1-R and by the
\emph{Generalitat Valenciana} under grant PROMETEO-II/2015/013
(SmartLogic).}
}

\author{Fred Mesnard\inst{1} \and \'Etienne Payet\inst{1} \and Germ\'an Vidal\inst{2}}

\institute{
  LIM - Universit\'e de la R\'eunion, France\\
  \email{\{frederic.mesnard,etienne.payet\}@univ-reunion.fr} 
\and 
MiST, DSIC, Universitat Polit\`ecnica de Val\`encia, Spain\\
  \email{gvidal@dsic.upv.es}
}

\maketitle

\pagestyle{plain}

\begin{abstract}
  Concolic testing is a popular dynamic validation technique that
  can be used for both model checking and automatic test case
  generation.
  We have recently introduced concolic testing in the context of logic
  programming. In contrast to previous approaches, the key ingredient
  in this setting is a technique to generate appropriate run-time
  goals by considering all possible ways an atom can unify with the
  heads of some program clauses. This is called ``selective''
  unification.
  In this paper, we show that the existing algorithm is not complete
  and explore different alternatives in order to have a sound and
  complete algorithm for selective unification.
\end{abstract}

\section{Introduction} \label{intro}

A popular approach to software validation is based on so called
\emph{concolic execution} \cite{GKS05,SMA05}, which combines both
\emph{conc}olic and symb\emph{olic} execution
\cite{Kin76,Cla76,APV09}.  Concolic \emph{testing} \cite{GKS05} is a
technique based on concolic execution for finding run time errors and
automatically generating test cases.  In this approach, both concrete
and symbolic executions are performed in parallel, so that concrete
executions may help to spot (run time) errors---thus avoiding false
positives---and symbolic executions are used to generate alternative
input data---new test cases---so that a good coverage is obtained.

In concolic testing of imperative programs, one should augment the
states with a so called \emph{path condition} that stores the
constraints on the variables of the symbolic execution. Then, after a
(possibly incomplete) concolic execution, these constraints are used
for producing alternative input data (e.g., by negating one of the
constraints). Furthermore, and this is one of the main advantages of
concolic testing over the original approach based solely on symbolic
execution, if the constraints in the path condition become too
complex, one can still take some values from the concrete execution to
simplify them. This is sound (but typically incomplete) and often
allows one to explore a larger execution space than just giving up (as
in the original approach based only on symbolic execution).  Some
successful tools that are based on concolic execution are, e.g., CUTE
\cite{SMA05}, SAGE \cite{GLM12}, and Java Pathfinder \cite{PR10}.

We have recently introduced concolic testing in the context of logic
programming \cite{MPV15}. There, a concolic state has the form
$\tuple{S\sep S'}$, where $S$ and $S'$ are sequences of concrete and
symbolic goals,\!\footnote{Following the linear semantics of
  \cite{SESGF11}, we consider sequences of goals to represent the
  leaves of the SLD tree built so far.} respectively. In logic
programming, the notion of \emph{symbolic} execution is very natural.
Indeed, the structure of both $S$ and $S'$ is the same---the sequences
of atoms have the same predicates and in the same order---and the only
difference is that some atoms might be less instantiated in $S'$ than
in $S$.

A key ingredient of concolic testing in logic programming is the
search for new concrete goals so that alternative paths can be
explored, thus improving the coverage achieved so far. Let us
illustrate it with an example. Consider the following (labelled) logic
program:
\[
\begin{array}{l@{~~~~~~}l@{~~~~~~}l}
  (\ell_1)~ p(s(a)). & (\ell_4)~ q(a).   & (\ell_6)~ r(a).  \\
  
  (\ell_2)~ p(s(W)) \leftarrow q(W). & (\ell_5)~ q(b). & (\ell_7)~ r(c). \\
    
  (\ell_3)~ p(f(X)) \leftarrow r(X). & \\
\end{array}
\]
Given the initial goal $p(f(a))$, a concolic execution would combine a
concrete execution of the form
\[
p(f(a)) \to_\id r(a) \to_\id true
\]
where $\id$ denotes the empty substitution, with another one for the more general goal $p(N)$:
\[
p(N) \to_{\{N/f(Y)\}} r(Y) \to_{\{Y/a\}} true
\]
that only mimicks the steps of the former derivation despite being
more general.
The technique in \cite{MPV15} would basically produce the following
concolic execution:
\[
\begin{array}{l@{}l@{}l@{}l@{}l@{}l}
  \tuple{p(f(a))_\id\sep p(N)_\id} & \leadsto_{c(\{\ell_3\},\{\ell_1,\ell_2,\ell_3\})}
                                   & \tuple{r(a)_\id &\sep & r(Y)_{\{N/f(Y)\}}} \\
                                   & \leadsto _{c(\{\ell_6\},\{\ell_6,\ell_7\})} &
                                   \tuple{\mathsf{true}_{\id}
                                     &\sep& \mathsf{true}_{\{N/f(a)\}}}
                                     \\
\end{array}
\]
where the goals are annotated with the answer computed so far. 
%
Roughly speaking, the above concolic execution 
is comprising the two standard SLD derivations for $p(f(a))$ and
$p(N)$ above. Moreover, it also includes some further information: 
the labels of the clauses that unified with each concrete and symbolic
goals.

For instance, the first step in the concolic execution above is
labelled with $c(\{\ell_3\},\{\ell_1,\ell_2,\ell_3\})$. This means
that the concrete goal only unified with clause $\ell_3$, but the
symbolic goal unified with clauses $\ell_1$, $\ell_2$ and
$\ell_3$. Therefore, when looking for new run time goals that explore
alternative paths, one should look for goals that unify with $\ell_1$
but not with $\ell_2$ and $\ell_3$, that unify with $\ell_1$ and
$\ell_2$ but not with $\ell_3$, and so forth. In general, we should
look for atoms that unify with all (and only) the feasible---i.e.,
those for which a solution exists---sets of clauses in
$\{\{\},\{\ell_1\},\{\ell_1,\ell_2\},\{\ell_1,\ell_2,\ell_3\},\{\ell_2\},\{\ell_2,\ell_3\}\}$.
Also, some additional constraints on the groundness of some arguments
are often required (e.g., to ensure that the generated goals are valid
\emph{run time} goals and, thus, will be terminating). A prototype
implementation of the concolic testing scheme for pure Prolog, called
\textsf{contest}, is publicly available from
\texttt{http://kaz.dsic.upv.es/contest.html}.

In this paper, we focus on the so called \emph{selective unification}
problem that must be solved in order to produce the alternative
goals during concolic testing. To be more precise, a selective
unification problem is determined by a tuple $\tuple{A,\Hpos,\Hneg,G}$
where
\begin{itemize}
\item $A$ is the selected atom in a symbolic goal, e.g., $p(N)$, 
\item $\Hpos$ are the atoms in the heads of the clauses we want $A$ to
  unify with, e.g., for $\{\ell_1,\ell_2\}$ in the example above, we
  have $\Hpos = \{ p(s(a)), p(s(W))\}$,
\item $\Hneg$ are the atoms in the heads of the clauses we do not want
  $A$ to unify with, e.g., for $\{\ell_1,\ell_2\}$ in the example
  above, we have $\Hneg=\{p(f(X))\}$,
\item $G$ is a set with the variables we want to be ground, e.g.,
  $\{N\}$.
\end{itemize}
In this case, the problem is satisfiable and a solution is
$\{N/s(a)\}$ since then $p(s(a))$ will unify with both atoms,
$p(s(a))$ and $p(s(W))$, but it will not unify with $p(f(X))$ and,
moreover, the variable $N$ is ground.

In contrast, the case $\{\ell_1\}$ is not feasible, since there is no
ground instance of $p(N)$ such that it unifies with $p(s(a))$ but not
with $p(s(W))$.

In \cite{MPV15}, we introduced a first algorithm for selective
unification. Unfortunately, this algorithm was incomplete. In this
paper, we further analyze this problem, identifying the potential
sources of incompleteness, proving a number of properties, and
introducing refined algorithms which are sound and complete under some
circumstances.

This paper is organized as follows. After some preliminaries in
Section~\ref{prelim}, Section~\ref{sec:sup} recalls and then extends
some of the developments in \cite{MPV15}. Then,
Section~\ref{sec:improving} introduces refined versions of the
algorithm for which we can obtain stronger results. Finally,
Section~\ref{future} concludes and points out several possibilities
for future work.

\section{Preliminaries} \label{prelim}

We assume some familiarity with the standard definitions and notations
for logic programs as introduced in \cite{Apt97}. Nevertheless, in
order to make the paper as self-contained as possible, we present in
this section the main concepts which are needed to understand our
development.

We denote by $|S|$ the cardinality of the set $S$.
In this work, we consider a first-order language with a fixed
vocabulary of predicate symbols, function symbols, and variables
denoted by $\Pi$, $\Sigma$ and $\cV$, respectively. 
%
%
We let $\cT(\Sigma,\cV)$ denote the set of \emph{terms} constructed
using symbols from $\Sigma$ and variables from $\cV$.  
Positions are used to address the nodes of a term viewed as a tree. A
\emph{position} $p$ in a term $t$, in symbols $p\in\pos(t)$, is
represented by a finite sequence of natural numbers, where $\toppos$
denotes the root position.
We let $t|_p$ denote the \emph{subterm} of $t$ at position $p$ and
$t[s]_p$ the result of \emph{replacing the subterm} $t|_p$ by the term
$s$. 
The depth $\depth(t)$ of a term $t$ is defined as: $\depth(t) = 0$ if
$t$ is a variable and $\depth(f(t_1,\ldots,t_n)) =
1+\mathsf{max}(\depth(t_1),\ldots,\depth(t_n))$, otherwise.
We say that $t|_p$ is a subterm of $t$ at depth $k$ if there are $k$
nested function symbols from the root of $t$ to the root of $t|_p$.
An \emph{atom} has the form $p(t_1,\ldots,t_n)$ with $p/n \in \Pi$ and $t_i
\in \cT(\Sigma,\cV)$ for $i = 1,\ldots,n$.  
The notion of position is extended to atoms in the natural way.  
A \emph{goal} is a finite sequence of atoms $A_1,\ldots,A_n$, where
the \emph{empty goal} is denoted by $true$.  A \emph{clause} has the
form $H \ot \cB$ where $H$ is an atom and $\cB$ is a goal (note that
we only consider \emph{definite} programs).  A logic \emph{program} is
a finite sequence of clauses.
%
$\var(s)$ denotes the set of variables in the syntactic object $s$
(i.e., $s$ can be a term, an atom, a query, or a clause).  A syntactic
object $s$ is \emph{ground} if $\var(s)=\emptyset$. In this work, we
only consider \emph{finite} ground terms.

Substitutions and their operations are defined as usual. In
particular, the set $\dom(\sigma) = \{x \in \cV \mid \sigma(x) \neq
x\}$ is called the \emph{domain} of a substitution $\sigma$. We let
$\id$ denote the empty substitution. The application of a substitution
$\theta$ to a syntactic object $s$ is usually denoted by
juxtaposition, i.e., we write $s\theta$ rather than $\theta(s)$. 
The \emph{restriction} $\theta\!\res_V$ of a substitution $\theta$ to a
set of variables $V$ is defined as follows: $x\theta\!\res_{V} =
x\theta$ if $x\in V$ and $x\theta\!\res_V = x$ otherwise. We say that
$\theta = \sigma~[V]$ if $\theta\!\res_V = \sigma\!\res_V$.
A syntactic object $s_1$ is \emph{more general} than a syntactic
object $s_2$, denoted $s_1 \sleq s_2$, if there exists a substitution
$\theta$ such that $s_2 = s_1\theta$. A \emph{variable renaming} is a
substitution that is a bijection on $\cV$. Two syntactic objects $t_1$
and $t_2$ are \emph{variants} (or equal up to variable renaming),
denoted $t_1 \sim t_2$, if $t_1 = t_2\rho$ for some variable renaming
$\rho$. A substitution $\theta$ is a unifier of two syntactic objects
$t_1$ and $t_2$ iff $t_1\theta = t_2\theta$; furthermore, $\theta$ is
the \emph{most general unifier} of $t_1$ and $t_2$, denoted by
$\mgu(t_1,t_2)$ if, for every other unifier $\sigma$ of $t_1$ and
$t_2$, we have that $\theta \sleq \sigma$. We write $t_1\approx t_2$
to denote that $t_1$ and $t_2$ unify for some substitution, which is
not relevant here. By abuse of notation, we also use $\mgu$ to denote
the most general unifier of a conjunction of equations of the form
$s_1=t_1\wedge\ldots\wedge s_n=t_n$, i.e.,
$\mgu(s_1=t_1\wedge\ldots\wedge s_n=t_n)=\theta$ if
$s_i\theta=t_i\theta$ for all $i=1,\ldots,n$ and for every other
unifier $\sigma$ of $s_i$ and $t_i$, $i=1,\ldots,n$, we have that
$\theta \sleq \sigma$.

We say that a syntactic object $o$ is \emph{linear} if it does not
contain multiple occurrences of the same variable. A substitution
$\{X_1/t_1,\ldots,X_n/t_n\}$ is \emph{linear} if $t_1,\ldots,t_n$ are
linear and, moreover, they do not share variables.

\section{The Selective Unification Problem} \label{sec:sup}

In this section, we first recall the unification problem from
\cite{MPV15}. There, an algorithm for ``selective unification'' was
proposed, and it was conjectured to be complete. Here, we prove that
it is indeed incomplete and we identify 
two sources of incompleteness.

\begin{definition}[selective unification problem] \label{def:sup} Let
  $A$ be an atom with $G\subseteq\var(A)$ a set of variables, and let
  $\Hpos$ and $\Hneg$ be finite sets of atoms such that all atoms are
  pairwise variable disjoint and $A\approx B$ for all
  $B\in\Hpos\cup\Hneg$.
  %
  Then, the \emph{selective unification problem} for $A$ w.r.t.\
  $\Hpos$, $\Hneg$ and $G$ is defined as follows:
  \[
  \cP(A,\Hpos,\Hneg,G) = \left\{\sigma\!\!\res_{\var(A)} ~\begin{array}{|@{~}ll}
                                \hspace{2.3ex}\forall H\in\Hpos: A\sigma\approx
                                H \\
                                \wedge ~ \forall H\in\Hneg: \neg
                                (A\sigma\approx H) \\
                                \wedge ~ G\sigma ~\mbox{is ground} \\
                                \end{array}\right\}
  \]
\end{definition}
%
%
When the considered signature is finite, the following algorithm is
sound and complete for solving the selective unification problem:
first, bind the variables of $A$ with terms of depth $0$. If the
condition above does not hold, then we try with terms of depth $1$,
and check it again. We keep increasing the considered term depth until
a solution is found. Moreover, there exists a finite number $n$ such
that, if a solution has not been found when considering terms of depth
$n$, then the problem is not satisfiable.

\begin{theorem}
  Let $A$ be a linear atom with $G\subseteq\var(A)$, $\Hpos$ be a
  finite set of linear atoms and $\Hneg$ be a finite set of atoms such
  that all atoms are pairwise variable disjoint and $A\approx B$ for
  all $B\in\Hpos\cup\Hneg$. Then, checking that
  $\cP(A,\Hpos,\Hneg,G)\neq\emptyset$ is decidable.
\end{theorem}

\begin{proof}
  Here, we assume the naive algorithm sketched above.
  Let us first consider that all atoms in $\{A\}\cup\Hpos\cup\Hneg$ are
  linear. Let $k$ be the maximum depth of the atoms in
  $\{A\}\cup\Hpos\cup\Hneg$.  Consider the set 
  \[
  \Theta' = \{ \theta\mid\dom(\theta)\subseteq\var(A),~\depth(A\theta)\sleq
  k+1\}
  \]
  On $\Theta'$, we define the binary relation 
  $\theta_1 \simeq \theta_2$ iff $A\theta_1 \sim A \theta_2$.
  The relation  $\simeq$ is an equivalence relation.
  Let $\Theta=\Theta' /\!\!\simeq$. 
  The set $\Theta$ is usually large but 
  finite. Now, we proceed by contradiction and assume that the problem is satisfiable
  but there is no solution in $\Theta$. 

  Let $\sigma\in \cP(A,\Hpos,\Hneg,G)$ be one of such solutions with
  $\sigma\not\in\Theta$. Let $k'\sleq k$ be the maximum depth of the
  atoms in $\Hpos$. Let $s_1,\ldots,s_n$ be the non-variable terms at
  depth $k'+1$ or higher in $A\sigma$, which occur at positions
  $p_1,\ldots,p_n$.  Trivially, all atoms in $\Hpos$ should have a
  variable at depth $k'$ or lesser in order to still unify with
  $A\sigma$.  Therefore, replacing these terms by any term would not
  change the fact that it unifies with all atoms in $\Hpos$. Formally,
  $(\ldots(A\sigma[t_1]_{p_1})\ldots)[t_n]_{p_n} \approx H$ for all
  $H\in\Hpos$ and for all terms $t_1,\ldots,t_n$.

  Now, let us consider the negative atoms $\Hneg$. Let us focus in the
  worst case, where the maximum depth of the atoms in $\Hneg$ is
  $k\geq k'$. Since $\neg(A\sigma\approx H)$ for all $H\in\Hneg$ and
  $(\ldots(A\sigma[t_1]_{p_1})\ldots)[t_n]_{p_n} \approx H$ for all
  $H\in\Hpos$ and for all terms $t_1,\ldots,t_n$, let us choose terms
  $t'_1,\ldots,t'_n$ such that
  $\neg((\ldots(A\sigma[t'_1]_{p_1})\ldots)[t'_n]_{p_n} \approx H)$
  for all $H\in\Hneg$ and
  $(\ldots(A\sigma[t'_1]_{p_1})\ldots)[t'_n]_{p_n}$ has depth
  $k+1$. Note that this is always possible since, in the worst case,
  for each term in the atoms of $\Hneg$ at depth $k$, we might need a
  term at depth $k+1$ (when the term in the atom of $\Hneg$ is the
  only constant of the signature, so we need to introduce a function
  symbol and another constant if the argument should be ground).
  Let $\sigma'\subseteq\dom(A)$ be a subtitution such that $A\sigma' =
  (\ldots(A\sigma[t'_1]_{p_1})\ldots)[t'_n]_{p_n}$. Then, $\sigma' \in
  \cP(A,\Hpos,\Hneg,G)$ with $\sigma'\in\Theta$ and, thus, we get a
  contradiction.


  Extending the proof to non-linear atoms is not difficult but it is
  tedious since we have to consider a higher depth that may depend on
  the multiple occurrences of the same variables.  \qed
\end{proof}
We conjecture that the above naive algorithm would
also be complete for infinite signatures (e.g., integers) since the
number of symbols in the considered atoms is finite.
Nonetheless, such algorithms may be so inefficient that they are
impractical in the context of concolic testing.

We note that the set $\cP(A,\Hpos,\Hneg,G)$ is usually
infinite. Moreover, even when considering only the \emph{most general}
solutions in this set, there may still exist more than one:

\begin{example} \label{ex1} Consider $A=p(X,Y)$,
  $\Hpos=\{p(Z,Z),p(a,b)\}$, $\Hneg=\{p(c,c)\}$ and
  $G=\emptyset$. Then, both substitutions $\{X/a,Y/U\}$ and
  $\{X/U,Y/b\}$ are most general solutions in $\cP(A,\Hpos,\Hneg,G)$.
  In principle, any of them is equally good in our context.
\end{example}
In \cite{MPV15}, we have introduced a stepwise method that,
intuitively speaking, proceeds as follows:
\begin{itemize}
\item First, we produce some ``maximal'' substitutions $\theta$ for
  $A$ such that $A\theta$ still unifies with the atoms in
  $\Hpos$. Here, we use a special set $\cU$ of fresh variables with
  $\var(\{A\}\cup\Hpos\cup\Hneg)\cap\cU=\emptyset$. The elements of
  $\cU$ are denoted by $U$, $U'$, $U_1$\dots{} Then, in $\theta$, the
  variables from $\cU$ (if any) denote positions where further binding
  \blue{\emph{might}} prevent $A\theta$ from unifying with some atom
  in $\Hpos$. 
  %
  %

\item In a second stage, we look for another substitution $\eta$ such
  that $\theta\eta$ is a solution of the selective unification
  problem, i.e., $\theta\eta \in\cP(A,\Hpos,\Hneg,G)$. Here, we
  basically follow a generate and test algorithm (as in the naive
  algorithm above), but it is now more restricted thanks to the
  bindings in $\theta$ and the fact that binding variables from $\cU$
  is not allowed.
\end{itemize}
%
%
In the first stage, we use the variables from the special set $\cU$ to
replace \emph{disagreement pairs} (see~\cite{Apt97} p.~27).
Roughly speaking, given terms $s$ and $t$, a subterm $s'$ of $s$ and a
subterm $t'$ of $t$ form a disagreement pair if the root symbols of
$s'$ and $t'$ are different, but the symbols from $s'$ up to the root of
$s$ and from $t'$ up to the root of $t$ are the same. For instance,
$X,g(a)$ and $b,h(Y)$ are disagreement pairs of the terms $f(X,g(b))$
and $f(g(a),g(h(Y)))$.
A disagreement pair $t,t'$ is called \emph{simple}
if one of the terms is a variable that does not occur in the other
term and no variable of $\cU$ occurs in $t,t'$.  We say that
the substitution $\{X/s\}$ is determined by $t,t'$ if
$\{X,s\}=\{t,t'\}$.


\begin{definition}[algorithm for positive unification] \label{alg1}
\begin{description}
\item[\textbf{Input:}] 
  an atom $A$ and a set of atoms $\Hpos$ such that all atoms are
  pairwise variable disjoint and $A\approx B$ for all $B\in\Hpos$.
\item[\textbf{Output:}] a substitution $\theta$.
\end{description}

\begin{enumerate}
\item \label{algo-msa-init}
  Let $\cB:=\{A\}\cup\Hpos$.
\item \label{algo-msa-while-simple}
  While simple disagreement pairs occur in $\cB$ do
  \begin{enumerate}
  \item nondeterministically choose a simple disagreement pair $X,t$
    (respectively, $t,X$) in $\cB$;
  \item \label{algo-msa-simple-pair}
    set $\cB$ to $\cB\eta$ where $\eta = \{X/t\}$.
  \end{enumerate}
\item \label{algo-msa-while-not-simple}
  While $|\cB|\neq 1$ do
  \begin{enumerate}
  \item nondeterministically choose a 
    disagreement pair $t,t'$ in $\cB$;
  \item \label{algo-msa-not-simple-pair} replace 
    $t,t'$ 
    with a fresh variable from $\cU$.
  \end{enumerate}
\item \label{algo-msa-return} Return $\theta\blue{\gamma}$, where
  $\cB=\{B\}$, $A\theta = B$, $\Dom(\theta)\subseteq\Var(A)$,
  \blue{and $\gamma$ is a variable renaming for the variables of
    $\var(A\theta)\backslash\cU$ with fresh variables from
    $\cV\backslash\cU$}.
\end{enumerate}
We denote by $\iclpp(A,\Hpos)$ the set of non-deterministic
substitutions computed by the above algorithm.
\end{definition}
%
%
\blue{Observe that the step (2a) involves two types of non-determinism:
\begin{itemize}
\item \emph{Don't care} nondeterminism, when there are several
  disagreement pairs $X,t$ (or $t,X$) for \emph{different}
  variables. In this case, we can select any of them and continue with
  the next step. The final solution would be the same no matter the
  selection. This is also true for step (3a), since the order in which
  the non-simple disagreement pairs are selected will not affect the
  final result.
\item \emph{Don't know} nondeterminism, when there are several
  disagreement pairs $X,t$ (or $t,X$) for the same variable $X$. In
  this case, we should consider all possibilities since they may give
  rise to different solutions.
\end{itemize}
}%

\begin{example}\label{ex-1-g}
  Let $A=p(X,Y)$ and $\Hpos=\{p(a,b),p(Z,Z)\}$. Therefore, we start
  with $\cB:=\{p(X,Y),p(a,b),p(Z,Z)\}$.  The algorithm then considers
  the simple disagreement pairs in $\cB$. From $X,a$, we get
  $\eta_1:=\{X/a\}$ and the action~(\ref{algo-msa-simple-pair}) sets
  $\cB$ to $\cB\eta_1=\{p(a,Y),p(a,b),p(Z,Z)\}$.  The substitution
  $\eta_2:=\{Y/b\}$ is determined by $Y,b$ and the
  action~(\ref{algo-msa-simple-pair}) sets $\cB$ to
  $\cB\eta_2=\{p(a,b),p(Z,Z)\}$.  Now, we have two don't know nondeterministic
  possibilities:
   \begin{itemize}
   \item If we consider the disagreement pair $a,Z$, we have a
     substitution $\eta_3:=\{Z/a\}$ and
     action~(\ref{algo-msa-simple-pair}) then sets $\cB$ to
     $\cB\eta_3=\{p(a,b),p(a,a)\}$. Now, no simple disagreement pair
     occurs in $\cB$, hence the algorithm jumps to the loop at
     line~\ref{algo-msa-while-not-simple}.
     Action~(\ref{algo-msa-not-simple-pair}) replaces the disagreement
     pair $b,a$ with a fresh variable $U\in \cU$, hence $\cB$ is set
     to $\{p(a,U)\}$.  As $|\cB|=1$ the loop at
     line~\ref{algo-msa-while-not-simple} stops and the algorithm
     returns the substitution $\{X/a,Y/U\}$.
   \item If we consider the disagreement pair $b,Z$ instead, we have  a
     substitution $\eta'_3:=\{Z/b\}$, and
     action~(\ref{algo-msa-simple-pair}) sets $\cB$ to
     $\cB\eta'_3=\{p(a,b),p(b,b)\}$. Now, by proceeding as in the
     previous case, the algorithm returns 
     $\{X/U',Y/b\}$.
   \end{itemize}
   Therefore, $\iclpp(A,\Hpos) = \{\{X/a,Y/U\},\{X/U',Y/b\}\}$.
\end{example}
%
%
The soundness of the algorithm in Definition~\ref{alg1} can then be
proved as follows (termination is straightforward, see
\cite{MPV15corr}).  Note that this result was incomplete in
\cite{MPV15} since the condition on $\Ran(\eta)$ was missing.

\begin{theorem}\label{theorem:correction-algo-pos}
  Let $A$ be an atom and $\Hpos$ be a set of atoms such that all atoms
  are pairwise variable disjoint and $A\approx B$ for all $B\in\Hpos$.
  Then, for all $\theta\in\iclpp(A,\Hpos)$, we have that $A\theta\eta
  \approx H$ for all $H\in\Hpos$ and for any \blue{idempotent}
  substitution $\eta$ with
  $\Dom(\eta)\subseteq\Var(A\theta)\backslash\cU$ \blue{and
    $\Ran(\eta)\cap(\var(\Hpos\cup\{A\})\cup\cU)=\emptyset$}.
\end{theorem}
In order to prove this theorem, we first need the following results,
which can be found in \cite[Appendix B.2]{MPV15corr}:

\begin{lemma}\label{lemma:technical-1}
  Suppose that $A\theta=B\theta$ for some atoms $A$ and $B$
  and some substitution $\theta$.
  Then we have $A\theta\eta=B\eta\theta\eta$ for any
  substitution $\eta$ with
  $[\Dom(\eta)\cap \Var(B)]\cap\Dom(\theta)=\emptyset$
  and
  $Ran(\eta)\cap\Dom(\theta\eta)=\emptyset$.
\end{lemma}

\begin{proposition}\label{proposition:invariant-correction-algo-pos-2}
  The loop at line~\ref{algo-msa-while-not-simple} always terminates
  and the following statement is an invariant of this loop.
  \begin{description}
  \item[$\mathrm{(inv')}$] For each $A'\in\{A\}\cup\Hpos$ there exists 
    $B\in\cB$ and a substitution $\theta$ such that
    $A'\theta=B\theta$, \blue{$\dom(\theta)\subseteq (\var(\Hpos\cup\{A\})\cup\cU)$} and
    $\Var(\cB)\cap\Dom(\theta)
    \subseteq \cU$.
  \end{description}
\end{proposition}
%
%
%
The proof of Theorem~\ref{theorem:correction-algo-pos} can now proceed
as follows:

\begin{proof}
  Upon termination of the loop at line~\ref{algo-msa-while-not-simple}
  we have $|\cB|=1$. Let $B$ be the element of $\cB$ with $A\theta=B$,
  and let $\theta'\in\iclpp(A,\Hpos)$ be a renaming of $\theta$ for
  the variables of $A\theta\backslash\cU$. 
  By Proposition~\ref{proposition:invariant-correction-algo-pos-2}, we
  have that, for all $H\in\Hpos$, there exists a substitution $\mu$
  such that $A\theta\mu= H\mu$ and the following conditions hold:
  \begin{itemize}
  \item $\dom(\mu)\subseteq(\var(\Hpos\cup\{A\})\cup\cU)$ and
  \item $\Var(A\theta)\cap\Dom(\mu) \subseteq \cU$.
  \end{itemize}
  \blue{
    Trivially, there exists a unifier $\mu'$ for $A\theta'$ and $H$
    too, and the same conditions hold:
    $\dom(\mu')\subseteq(\var(\Hpos\cup\{A\})\cup\cU)$ and
    $\Var(A\theta')\cap\Dom(\mu') \subseteq \cU$.
  
  Now, in order to apply Lemma~\ref{lemma:technical-1}, we need to
  prove the following conditions: 
  \begin{itemize}
  \item $[\dom(\eta)\cap\var(A\theta')]\cap\dom(\mu') =
    \emptyset$. This is trivially implied by the fact that
    $\dom(\eta)\subseteq\var(A\theta')\backslash\cU$ and
    $\Var(A\theta')\cap\Dom(\mu') \subseteq \cU$.

  \item $\Ran(\eta)\cap\Dom(\mu'\eta) = \emptyset$. First, since
    $\Dom(\mu'\eta)\subseteq\dom(\mu')\cup\dom(\eta)$, we prove the
    stronger claim: $\Ran(\eta)\cap\Dom(\mu') = \emptyset$ and
    $\Ran(\eta)\cap\Dom(\eta) = \emptyset$. The second condition is
    triviallly implied by the idempotency of $\eta$. Regarding the
    first condition, it is implied by
    $\Ran(\eta)\cap(\var(\Hpos\cup\{A\})\cup\cU) = \emptyset$ since
    $\dom(\mu')\subseteq(\var(\Hpos\cup\{A\})\cup\cU)$, which is true.
  \end{itemize}
  Therefore, by Lemma~\ref{lemma:technical-1}, we have that
  $A\theta'\eta\mu'\eta = H\mu'\eta$ and, thus, $A\theta'\eta$ unifies
  with $H$. Hence, we have proved that $A\theta'\eta$ unifies with
  every atom in $\Hpos$. \qed
}
\end{proof}
Now we deal with the negative atoms and the groundness constraints by
means of the following algorithm:


\begin{definition}[algorithm for selective unification] \label{alg2}
\begin{description}
\item[\textbf{Input:}] an atom $A$ with $G\subseteq\var(A)$ a set of
  variables, and two finite sets $\Hpos$ and $\Hneg$ such that all
  atoms are pairwise variable disjoint and $A\approx B$ for all
  $B\in\Hpos\cup\Hneg$.
\item[\textbf{Output:}] $\fail$ or a substitution $\theta\eta$
  (restricted to the variables of $A$).
\end{description}
\begin{enumerate}
\item \blue{Generate---using a fair algorithm---}pairs
  $(\theta,\eta)$ with $\theta\in\iclpp(A,\Hpos)$ and $\eta$ an
  \blue{idempotent} substitution such that $G\theta\eta$ is ground,
  $\Dom(\eta)\subseteq\Var(A\theta)\backslash\cU$ \blue{and
    $\Ran(\eta)\cap(\var(\Hpos\cup\{A\})\cup\cU)=\emptyset$},
  otherwise return $\fail$.
\item Check that for each $H^-\in\Hneg$, 
  $\neg (A\theta\eta\approx H^-)$,
  otherwise return $\fail$.
\item Return $\theta\eta\gamma$ (restricted to the variables of $A$),
  \blue{where $\gamma$ is a variable renaming for $A\theta\eta$ with
    fresh variables from $\cV\backslash\cU$}.
\end{enumerate}
We denote by $\iclp(A,\Hpos,\Hneg,G)$ the set of non-deterministic
(non-failing) substitutions computed by the above algorithm.
\end{definition}
\blue{Note that step (1) above is don't know nondeterministic and,
  thus, all substitutions in $\iclpp(A,\Hpos)$ should in principle be
  considered. On the other hand, computing the first solution of the
  above algorithm is enough for concolic testing.}

\blue{The soundness of the selective unification algorithm is a
  straightforward consequence of 
  Theorem~\ref{theorem:correction-algo-pos} and the fact that the
  algorithm in Definition~\ref{alg2} is basically a
  fair generate-and-test procedure.}

Unfortunately, the selective unification algorithm is not complete in general, as
Examples~\ref{ex:incomplete} and \ref{ex:incomplete2}
below illustrate.
Example~\ref{ex:incomplete} shows that the algorithm cannot
always compute all the solutions while
Example~\ref{ex:incomplete2} shows that it may even find no
solution at all for a satisfiable instance of the problem.

\begin{example}\label{ex:incomplete}
  Consider the atom $A= p(X_1,X_2)$ with $G= \{X_1\}$, and
  the sets $\Hpos = \{p(X,g(X)),p(Z,Z)\}$ and $\Hneg =
  \{p(g(b),W)\}$.
  Here, we have 
  \[
  \iclpp(A,\Hpos) =
  \{\underbrace{\{X_1/X',X_2/U\}}_{\theta_1},\underbrace{\{X_1/U,X_2/g(X')\}}_{\theta_2}\}
  \]
  %
  %
  The algorithm is able to compute the solution
  $\{X_1/g(a),X_2/U\}$ from $\theta_1$, $\eta=\{X'/g(a)\}$ and
  $\gamma=\id$.
  However,
  it cannot compute $\{X_1/g(a),X_2/g(X')\}\in\cP(A,\Hpos,\Hneg,G)$.
\end{example}
\blue{The algorithm fails here because the instantiation of
  variables from $\cU$ is not allowed. In \cite{MPV15}, it was
  incorrectly assumed that \emph{any} binding of a variable from $\cU$
  will result in a substitution $\theta'$ such that $A\theta'$ does not
  unify will all the atoms in $\Hpos$ anymore. However, the universal
  quantification was not right. For each variable from $\cU$, we can
  only ensure that \emph{there exists some} term $t$ such that binding
  this variable to $t$ will result in a substitution that prevents $A$
  from unifying with some atom in $\Hpos$.
  Therefore, since the algorithm of Definition~\ref{alg2} forbids the
  bindings of the variables in $\cU$, completeness is lost. We will
  propose a solution to this problem in the next section}


\begin{example} \label{ex:incomplete2}
  Consider $A=p(X_1,X_2)$,
  $\Hpos=\{p(X,a),p(b,Y)\}$, $\Hneg=\{p(b,a)\}$, and
  $G=\emptyset$.
  Here, we have $\iclpp(A,\Hpos) = \{\{X_1/b,X_2/a\}\}$ and, thus, the
  algorithm in Definition~\ref{alg2} fails. However, the following
  substitution $\{X_1/Z,X_2/Z\}$ is a solution, i.e.,
  $\{X_1/Z,X_2/Z\}\in\cP(A,\Hpos,\Hneg,G)$.
\end{example}
\blue{Unfortunately, we do not know how to generate such non-linear
  solutions except with the naive semi-algorithm mentioned at the
  beginning of this section, which is not generally useful in
  practice.  Therefore, in the next section we will rule out these
  solutions.}

\section{Recovering Completeness for Linear Selective Unification} \label{sec:improving}

In this section, we introduce different alternatives to recover the
completeness of the selective unification algorithm.


In the following, we only consider a subset of the solutions to the
selective unification problem, namely those which are \emph{linear}:
\[
\cP_\lin(A,\Hpos,\Hneg,G) = \{ \sigma \in \cP(A,\Hpos,\Hneg,G) \mid \sigma
~\mbox{is linear}\}
\]
i.e., we rule out solutions like those in Example~\ref{ex:incomplete2}
since we do not know how such solutions can be produced using a
constructive algorithm.
We refer to $\cP_\lin(A,\Hpos,\Hneg,G)$ as the \emph{linear selective
  unification problem}.

\subsection{A Naive Extension}
\label{sec:naive}

One of the sources of incompleteness of the algorithm in
Definition~\ref{alg2} comes from the fact that the variables from
$\cU$ cannot be bound. Therefore, one can consider a naive extension
of this algorithm as follows:

\begin{definition}[extended algorithm for selective unification] \label{alg2bis}
\begin{description}
\item[\textbf{Input:}] an atom $A$ with $G\subseteq\var(A)$ a set of
  variables, and two finite sets $\Hpos$ and $\Hneg$ such that all
  atoms are pairwise variable disjoint and $A\approx B$ for all
  $B\in\Hpos\cup\Hneg$.
\item[\textbf{Output:}] $\fail$ or a substitution $\theta\eta$
  (restricted to the variables of $A$).
\end{description}
\begin{enumerate}
\item Generate---using a fair algorithm---pairs $(\theta,\eta)$ with
  $\theta\in\iclpp(A,\Hpos)$ and $\eta$ an idempotent
  substitution such that $G\theta\eta$ is ground,
  \blue{$\Dom(\eta)\subseteq\Var(A\theta)$} and  
    $\Ran(\eta)\cap(\var(\Hpos\cup\{A\})\cup\cU)=\emptyset$,
  otherwise return $\fail$.  
\item Check that for each $H^-\in\Hneg$, 
  $\neg (A\theta\eta\approx H^-)$,
  otherwise return $\fail$.
\item Return $\theta\eta\gamma$ (restricted to the variables of $A$),
  where $\gamma$ is a variable renaming for $A\theta\eta$ with
    fresh variables from $\cV\backslash\cU$.
\end{enumerate}
We denote by $\iclp^\ast(A,\Hpos,\Hneg,G)$ the set of non-deterministic
(non-failing) substitutions computed by the above algorithm.
\end{definition}
In general, though, the above algorithm can be very inefficient since
all variables in $A\theta$ can now be bound, even those in
$\cU$. Nevertheless, one can easily define a fair procedure for
generating pairs $(\theta,\eta)$ in step (1) above which gives
priority to binding the variables in $\var(A\theta)\backslash\cU$, so
that the variables from $\cU$ are only bound when no solution can be
found otherwise.

\subsection{The Positive Unification Problem}
\label{section-algo-all-pos}

Now, we introduce a more efficient instance of the
algorithm for linear selective unification which is sound and
complete when the atoms in $A$ and $\Hpos$ are linear.
Formally, we are concerned with the following
unification problem:

\begin{definition}[positive linear unification
  problem] \label{def:pos} Let $A$ be a linear atom and let $\Hpos$ be
  a finite set of linear atoms such that all atoms are pairwise
  variable disjoint and $A\approx B$ for all $B\in\Hpos$.
  Then, the \emph{positive linear unification problem} for $A$ w.r.t.\
  $\Hpos$ is defined as follows:
  \[
  \cP^+_\lin(A,\Hpos) = \{\sigma\!\!\res_{\var(A)} \mid (\forall H\in\Hpos: A\sigma\approx
  H)~\mbox{and $\sigma$ is linear} \}
  \]
\end{definition}
Note that we do not want to find a unifier between $A$ and \emph{all}
the atoms in $\Hpos$, but a substitution $\theta$ such that $A\theta$
still unifies with \emph{each} atom in $\Hpos$ independently. So this
problem is different from the usual unification problems found in the
literature.

Clearly, $|\cP^+_\lin(A,\Hpos)|\geq 1$ since the identity substitution
is always a solution to the positive linear unification
problem. In general, though, the set $\cP^+_\lin(A,\Hpos)$ is
infinite.

\begin{example} \label{ex:pos}
  Let us consider $A=p(X)$ and $\Hpos=\{p(f(Y)), p(f(g(Z))) \}$. Then,
  we have
  $\{\id,\{X/f(X')\},\{X/f(g(X'))\},\{X/f(g(a))\},\{X/f(g(f(X')))\},\ldots\}\}\subseteq\cP^+_\lin(A,\Hpos)$,
  which is clearly infinite.
\end{example}
Therefore, in the following, we restrict our interest to so called \emph{maximal}
solutions:

\begin{definition}[maximal solution] \label{def:maximal} Let $A$ be a
  linear atom and $\Hpos$ be a finite set of linear atoms such that
  all atoms are pairwise variable disjoint and $A\approx B$ for all
  $B\in\Hpos$.  We say that a substitution $\theta \in \pplus$ is
  \emph{maximal} when the following conditions hold:
  \begin{enumerate} 
  \item\label{def:maximal-cond1} for any \blue{idempotent}
    substitution $\gamma$ with
    $\dom(\gamma)\subseteq\Var(A\theta)\setminus \cU$ \blue{and
      $\Ran(\gamma)\cap(\var(\Hpos\cup\{A\})\cup\cU)=\emptyset$},
    $(\theta\gamma)\blue{\res_{\var(A)}}$ is still an element of
    $\pplus$,
  \item\label{def:maximal-cond2} for any variable $U\in\Var(A\theta)
    \cap \cU$, we have that $(\theta\{U/t\})\!\!\res_{\var(A)}$ is not
    an element of $\pplus$ anymore for all non-variable term $t$, and
  \item \label{def:maximal-cond3} for any $X/t\in\theta$ and for all
    non-variable term $t|_p$, replacing it by a non-variable term
    rooted by a different symbol will result in a substitution which
    is not an element of $\pplus$ anymore.
  \end{enumerate}
  We let $\mus(A,\Hpos)$ denote the set of maximal solutions in $\pplus$.
\end{definition}
Intuitively speaking, given a maximal solution $\theta$, the first
condition implies that $(\theta\gamma)\!\res_{\var(A)}$ is still a
solution of the positive linear unification problem as long as no
variables from $\cU$ are bound. The second and third conditions mean
that the rest of the symbols in $\theta$ cannot be changed, i.e.,
binding a variable from $\cU$ with a non-variable term or changing any
constant or function symbol by a different one, will always result in
a substitution which is not a solution of the positive linear
unification problem anymore.

\begin{example} \label{ex:pos2} Consider, e.g., $A = p(X_1,X_2)$ and
  $\Hpos = \{p(f(Y),a), p(f(g(Z)),b) \}$. Here, we have
  $\{X_1/X',X_2/X''\}\in\pplus$ but it is not a maximal solution,
  i.e., $\{X_1/X',X_2/X''\}\not\in\mus(A,\Hpos)$ since binding $X''$
  to, e.g., $a$, will result in a substitution which is not in
  $\pplus$.  In contrast,
  $\{X_1/f(g(Z')),X_2/U\}$ is a maximal solution. However, any
  substitution of the form $\{X_1/f(g(t)),X_2/U\}$ for any
  non-variable term $t$ is not a maximal solution since the third
  condition will not hold anymore (one can change the symbols
  introduced by $t$ and still get a solution in $\pplus$).
  The substitution $\{X_1/f(Y'),X_2/U\}$ is not a maximal
  solution as well since binding $Y'$ to, e.g., $a$, will result
  in a substitution which is not in $\pplus$, hence the first
  condition does not hold. 
  And the same applies to $\{X_1/f(U'),X_2/U\}$, which is not a
  maximal solution either since we can bind $U'$ to $g(X')$ and still
  get a substitution in $\pplus$. 
\end{example}
In contrast to $\cP^+_\lin(A,\Hpos)$, the set $\mus(A,\Hpos)$ is
finite, since it is bounded by the depth of the terms in $\Hpos$.
Actually, for linear atoms in $\{A\}\cup\Hpos$, there is only
\emph{one} maximal solution.

\begin{proposition} \label{prop:mus-finite} Let $A$ be a linear atom
  and $\Hpos$ be a finite set of linear atoms such that all atoms are
  pairwise variable disjoint and $A\approx B$ for all
  $B\in\Hpos$. Then, the set $\mus(A,\Hpos)$ is a singleton (up to
  variable renaming).
\end{proposition}

\begin{proof}
  \blue{ 


    We proceed by contradiction. Let us assume that there are two
    maximal solutions $\sigma,\theta\in \mus(A,\Hpos)$, where
    $X/s\in\sigma$ and $X/t\in\theta$ for some variable
    $X\in\var(A)$. Let us consider that $s$ and $t$ differ at position
    $p$ such that $s|_p$ and $t|_p$ are rooted by a different symbol.
    Now, we distinguish the following cases:
    \begin{itemize}
    \item If $s|_p$ and $t|_p$ are rooted by different constant or
      function symbols, we get a contradiction by condition (3) of
      maximal solution.
    \item If $s|_p$ is rooted by a constant or function symbol, while
      $t|_p$ is rooted by a variable from $\cU$ (or viceversa), we get
      a contradiction by condition (2) of maximal solution.
    \item If $s|_p$ is rooted by a constant or function symbol, while
      $t|_p$ is rooted by a variable from $\cV\backslash\cU$ (or
      viceversa), we get a contradiction either by condition (1) or
      (3) of maximal solution.
    \item Finally, if $s|_p$ is rooted by a variable from $\cU$, while
      $t|_p$ is rooted by a variable from $\cV\backslash\cU$ (or
      viceversa), we get a contradiction either by condition (1) or
      (2) of maximal solution.
    \end{itemize}
    Therefore, the set $\mus(A,\Hpos)$ is necessarily a singleton.
    \qed }
\end{proof}
Moreover, the following key property holds: a maximal solution can
always be \emph{completed} in order to get a solution to the linear
unification problem when it is satisfiable.
In order to prove this result, we need to recall the definition of
\emph{parallel composition} of substitutions, denoted by $\Uparrow$ in
\cite{Pal90}.

\begin{definition}[parallel composition \cite{Pal90}] \label{uparrow}
  Let $\theta_1$ and $\theta_2$ be two idempotent
  substitutions.  Then, we define $\Uparrow$ as follows:
  \[ 
  \theta_1 \Uparrow \theta_2 = \left\{
    \begin{array}{l@{~~}l}
      \mgu(\widehat{\theta}_1 \cand \widehat{\theta}_2) & \mbox{if $\widehat{\theta_1}\cand\widehat{\theta_2}$ has a solution (a unifier)}\\
      \mathit{fail} & \mbox{otherwise}
    \end{array}\right.
  \]
  where $\widehat{\theta}$ denotes the \emph{equational
    representation} of a substitution $\theta$, i.e., if $\theta =
  \{X_1/ t_1, \ldots, X_n/ t_n \}$ then $\widehat{\theta}
  = (X_1 = t_1\cand \cdots\cand X_n = t_n)$.
\end{definition}

\begin{proposition} \label{prop:completeness1} Let $A$ be a linear
  atom and $\Hpos$ be a finite set of linear atoms such that all atoms
  are pairwise variable disjoint and $A\approx B$ for all
  $B\in\Hpos$. Let $\theta\in\mus(A,\Hpos)$ be the maximal solution
  for $A$ and $\Hpos$. Then, if $\cP_\lin(A,\Hpos,\Hneg,G)$ is
  satisfiable (the set contains at least one substitution), then there
  exists some substitution $\gamma$ such that $\theta\gamma\in
  \cP_\lin(A,\Hpos,\Hneg,G)$.
\end{proposition}

\begin{proof}
  \blue{For simplicity, we consider that $A=p(X)$, 
    $\Hpos=\{p(t_1),\ldots,p(t_n)\}$ and
    $\Hneg=\{p(s_1),\ldots,p(s_m)\}$. 
    Since the atoms are linear, the
    claim would follow by a similar argument. Let $\theta =
    \{X/t\}\in\mus(A,\Hpos)$ be the maximal solution.
    Hence, we have $t\approx t_i$ for all $i=1,\ldots,n$.
    Let $\sigma \in \cP_\lin(A,\Hpos,\Hneg,G)$ be a solution to the
    selective unification problem. By definition of maximal solution,
    there may be other solutions to the positive unification problem,
    but every introduced symbol cannot be different if we want to
    still unify with all terms $t_1,\ldots,t_n$ by condition (3) in
    the definition of maximal solution. Therefore, both substitutions
    must be compatible, i.e., we have
    $\theta\Uparrow\sigma=\delta\neq\fail$. Furthermore, taking into
    account the negative atoms in $\Hneg$ as well as the groundness
    constraints w.r.t.\ $G$,
    $\delta$ can only introduce further bindings, but
    would never require to generalize any term introduced by $\theta$
    and, thus, $\delta$ can be decomposed as $\theta\gamma$, with
    $\theta\gamma\in \cP_\lin(A,\Hpos,\Hneg,G)$.\qed}
\end{proof}
Therefore, computing the maximal solution 
suffices to check for satisfiability. 
%
Here, we use again the algorithm in Definition~\ref{alg1} for
computing the maximal solution, with the following differences:
\begin{itemize}
\item First, both $A$ and the atoms in $\Hpos$ are now linear.
\item Moreover, step (2a) is now don't care nondeterministic, so the
  algorithm will return a single solution, which we denote by
  $\iclpp_\lin(A,\Hpos)$.
\end{itemize}
%

\begin{proposition} \label{prop:completeness2} Let $A$ be a linear
  atom and $\Hpos$ be a finite set of linear atoms such that all atoms
  are pairwise variable disjoint and $A\approx B$ for all
  $B\in\Hpos$. Then, $\iclpp_\lin(A,\Hpos)=\mus(A,\Hpos)$.
\end{proposition}

\begin{proof} (sketch)
  \blue{The fact that $\iclpp_\lin(A,\Hpos)$ returns a singleton is
    trivial by definition, since the algorithm has no don't know
    nondeterminism and no step admits a failure.

    Regarding the fact that $\theta$ is a maximal solution, let us
    prove that all three conditions in Definition~\ref{def:maximal}
    hold. The first condition of maximal solution follows by
    Theorem~\ref{theorem:correction-algo-pos}, which is proved for the
    more general case of arbitrary (possibly non-linear) atoms.
    The third condition holds from the fact that in step (2) of
    $\iclpp_\lin$ only symbols from the atoms $A$ and $\Hpos$ are
    introduced following a $mgu$-like algorithm; therefore they are
    possibly not necessary, but cannot be replaced by different
    symbols and still unify with all the atoms in $\Hpos$.
    Finally, the second condition derives from step (3) of
    $\iclpp_\lin$ where non-simple disagreement pairs are replaced by
    fresh variables from $\cU$ and, thus, any binding to a
    non-variable term would result in $A\theta$ not unifying with some
    atom of $\Hpos$.
%
    \qed}
\end{proof}






%

\subsection{Dealing with the Negative Atoms}
\label{section-algo-all-pos-neg}

The algorithm $\iclp$ in Definition~\ref{alg2} is now redefined as
follows:

\begin{definition}[algorithm for linear selective unification] \label{alg2lin}
\begin{description}
\item[\textbf{Input:}] a linear atom $A$ with $G\subseteq\var(A)$ a set of
  variables, and two finite sets $\Hpos$ and $\Hneg$ such that the
  atoms in $\Hpos$ are linear and all
  atoms are pairwise variable disjoint and $A\approx B$ for all
  $B\in\Hpos\cup\Hneg$.
\item[\textbf{Output:}] $\fail$ or a substitution $\theta\eta$
  (restricted to the variables of $A$).
\end{description}
\begin{enumerate}
\item \blue{Let $\{\theta\}=\iclpp_\lin(A,\Hpos)$. Then,
    generate---using a fair algorithm---linear idempotent substitutions
    $\eta$ such that $G\theta\eta$ is ground,
    $\Dom(\eta)\subseteq\Var(A\theta)\backslash\cU$ and
    $\Ran(\eta)\cap(\var(\Hpos\cup\{A\})\cup\cU)=\emptyset$, otherwise
    return $\fail$.}
\item Check that for each $H^-\in\Hneg$, 
  $\neg (A\theta\eta\approx H^-)$,
  otherwise return $\fail$.
\item Return $\theta\eta\gamma$ (restricted to the variables of $A$),
  where $\gamma$ is a variable renaming for $A\theta\eta$ with fresh
  variables from $\cV\backslash\cU$.
\end{enumerate}
We denote by $\iclp_\lin(A,\Hpos,\Hneg,G)$ the set of
non-deterministic (non-failing) substitutions computed by the above
algorithm.
\end{definition}
%

\begin{example} \label{ex:posneg} Consider again $A = p(X_1,X_2)$ and
  $\Hpos = \{p(f(Y),a), p(f(g(Z)),b) \}$, together with $\Hneg =
  \{p(f(g(a)),c)\}$ and $G=\{X_1\}$. The algorithm for linear positive
  unification returns the maximal substitution
  $\{X_1/f(g(Z')),X_2/U\}$. Therefore, the algorithm for linear
  selective unification would eventually produce a solution of the
  form $\theta=\{X_1/f(g(b)),X_2/X'\}$ since $A\theta = p(f(g(b),X')$
  unifies with $p(f(Y),a)$ and $p(f(g(Z)),b)$ but not with
  $p(f(g(a)),c)$ and, moreover, $X_1$ is not ground. 
  However, if we consider a non-maximal solution, the algorithm in
  Definition~\ref{alg2} may fail, even if there exists some solution
  to the linear selective unification problem. This is the case, e.g.,
  if we consider the non-maximal solution $\{X_1/f(g(a)),X_2/U\}$.
\end{example}

\begin{theorem}[soundness]
  Let $A$ be a linear atom with $G\subseteq\var(A)$, $\Hpos$ be a
  finite set of linear atoms and $\Hneg$ be a finite set of atoms such
  that all atoms are pairwise variable disjoint and $A\approx B$ for
  all $B\in\Hpos\cup\Hneg$. Then, we have
  $\iclp_\lin(A,\Hpos,\Hneg,G)\subseteq\cP_\lin(A,\Hpos,\Hneg,G)$.
\end{theorem}

\begin{proof}
  \blue{The claim follows from Proposition~\ref{prop:completeness2} by
    assuming that the don't know nondeterministic substitutions
    considered in step (1) of the algorithm of
    Definition~\ref{alg2lin} are obtained by a fair generate-and-test
    algorithm which produces substitutions systematically starting
    with terms of depth $0$, then depth $1$, etc., as in the
    naive algorithm described at the beginning of
    Section~\ref{sec:sup}.  }
\end{proof}
The following result states the completeness of our algorithm. In
principle, we do not guarantee that all solutions are computed using
our algorithms, even for the linear case. However, we can ensure that
if the linear selective unification problem is satisfiable, our
algorithm will find at least one solution.

\begin{theorem}[completeness]
  Let $A$ be a linear atom with $G\subseteq\var(A)$, $\Hpos$ be a
  finite set of linear atoms and $\Hneg$ be a finite set of atoms such
  that all atoms are pairwise variable disjoint and $A\approx B$ for
  all $B\in\Hpos\cup\Hneg$. Then, if
  $\cP_\lin(A,\Hpos,\Hneg,G)\neq\emptyset$ (i.e., it is satisfiable),
  then $\iclp_\lin(A,\Hpos,\Hneg,G)\neq\emptyset$.
\end{theorem}

\begin{proof}
  \blue{By Proposition~\ref{prop:completeness1}, if
    $\cP_\lin(A,\Hpos,\Hneg,G)\neq\emptyset$ and $\theta$ is the
    computed maximal solution, then there exists a substitution
    $\gamma$ such that $(\theta\gamma)\!\res_{\var(A)} \in
    \cP_\lin(A,\Hpos,\Hneg,G)$. Moreover, such a substitution $\gamma$
    can be obtained by a fair generate-and-test algorithm such as that
    considered in Definition~\ref{alg2lin}. Finally, the claim follows
    by Proposition~\ref{prop:completeness2} which ensures that the
    algorithm in Definition~\ref{alg1} will always produce the maximal
    solution for $A$ and $\Hpos$.}
\end{proof}
In general, though, we cannot ensure that all solutions are computed
(which is not a drawback of the algorithm since we are only interested
in finding one solution if it exists):

\begin{example}
  Consider again $A = p(X_1,X_2)$ and $\Hpos = \{p(f(Y),a),
  p(f(g(Z)),b) \}$, together with $\Hneg = \{p(g(W),c)\}$ and
  $G=\emptyset$. The algorithm for linear positive unification returns
  the maximal substitution $\{X_1/f(g(Z')),X_2/U\}$. Therefore, it is
  impossible that the algorithm in Definition~\ref{alg2} could produce
  a solution like $\{X_1/f(X'),X_2/X''\}\in\cP_\lin(A,\Hpos,\Hneg,G)$.
\end{example}

  

\section{Discussion} \label{future}

In this paper, we have studied the soundness and completeness of
selective unification, a relevant operation in the context of concolic
testing of logic programs. 
In contrast to \cite{MPV15}, we have provided a refined correctness
result (one condition was missing in \cite{MPV15}), and we have also
identified the main sources of incompleteness for the algorithm in
\cite{MPV15}. Then, we have introduced several refinements so that the
procedure is now sound and complete w.r.t.\ linear solutions.
We are not aware of any other work that deals with the kind of
unification problems that we consider in this paper.

Clearly, the fact that we only consider linear solutions (i.e., the
relation $\cP_\lin$) means that our procedure can be incomplete in
general.  For instance, we consider the problem shown in
Example~\ref{ex:incomplete2} unsatisfiable, though a nonlinear
solution exists. Nevertheless, we do not expect this restriction to
have a significant impact in practice and, moreover, concolic testing
algorithms are usually incomplete in order to avoid a state explosion.
On the other hand, the refined algorithm in
Sections~\ref{section-algo-all-pos} and \ref{section-algo-all-pos-neg}
only considers linear atoms. This restriction may have a more
significant impact since many programs have nonlinear atoms in the
heads of the clauses and/or equalities in the bodies. In such cases,
we can still resort to using the algorithm of Section~\ref{sec:naive},
though it may be less efficient.

As for future work, we are considering to introduce a technique to
``linearize'' the atoms in $A\cup\Hpos$ by introducing some
constraints which could be solved later in the algorithm (e.g.,
replacing $p(X,X)$ by $p(X,Y)$ and the constraint $X=Y$).

Another interesting line of research involves improving the efficiency
of the selective unification algorithm. For this purpose, we plan to
investigate the conditions ensuring the following property:
\[
\mbox{if}~ \cP_\lin(A,\Hpos,\Hneg,G)=\emptyset, ~\mbox{then}~
\cP_\lin(A\theta,\Hpos,\Hneg,G)=\emptyset~\mbox{for all substitution
$\theta$}
\]
If this property indeed holds, then one could check \emph{statically}
the satisfiability of all possible selective unification problems in a
program, e.g., for atoms of the form $p(X_1,\ldots,X_n)$. We can then
use this information during concolic testing to rule out those
problems which we know are unfeasible no matter the run time values
(denoted by $\theta$).  From our preliminary experience with the tool
\texttt{contest} (\texttt{http://kaz.dsic.upv.es/contest.html}), this
might result in significant efficiency improvements.

Finally, we are also considering the definition of a possibly
approximate formulation of selective unification which could be
solved using an SMT solver. This might imply a loss of completeness,
but will surely improve the efficiency of the process.  Moreover, it
will also allow a smoother integration with the constraint solving
process which is required when extending our concolic testing
technique to full Prolog programs.

\bibliographystyle{plain}
\bibliography{biblio}

\begin{thebibliography}{10}

\bibitem{APV09}
Saswat Anand, Corina~S. Pasareanu, and Willem Visser.
\newblock Symbolic execution with abstraction.
\newblock {\em STTT}, 11(1):53--67, 2009.

\bibitem{Apt97}
K.R. Apt.
\newblock {\em {From Logic Programming to Prolog}}.
\newblock Prentice Hall, 1997.

\bibitem{Cla76}
L.A. Clarke.
\newblock A program testing system.
\newblock In {\em Proceedings of the 1976 Annual Conference (ACM'76)}, pages
  488--491, 1976.

\bibitem{GKS05}
P.~Godefroid, N.~Klarlund, and K.~Sen.
\newblock {DART: directed automated random testing}.
\newblock In {\em Proc.\ of PLDI'05}, pages 213--223. ACM, 2005.

\bibitem{GLM12}
P.~Godefroid, M.Y. Levin, and D.A. Molnar.
\newblock Sage: whitebox fuzzing for security testing.
\newblock {\em CACM}, 55(3):40--44, 2012.

\bibitem{Kin76}
James~C. King.
\newblock Symbolic execution and program testing.
\newblock {\em CACM}, 19(7):385--394, 1976.

\bibitem{MPV15}
F.~Mesnard, {\'{E}}.~Payet, and G.~Vidal.
\newblock Concolic testing in logic programming.
\newblock {\em {TPLP}}, 15(4-5):711--725, 2015.

\bibitem{MPV15corr}
F.~Mesnard, {\'{E}}.~Payet, and G.~Vidal.
\newblock Concolic testing in logic programming (extended version).
\newblock {\em CoRR}, abs/1507.05454, 2015.
\newblock Available from the following URL:
  \texttt{http://arxiv.org/abs/1507.05454}.

\bibitem{Pal90}
C.~Palamidessi.
\newblock Algebraic {P}roperties of {I}dempotent {S}ubstitutions.
\newblock In M.S. Paterson, editor, {\em Proc. of 17th Int'l Colloquium on
  Automata, Languages and Programming}, pages 386--399. Springer LNCS 443,
  1990.

\bibitem{PR10}
C.S. Pasareanu and N.~Rungta.
\newblock {Symbolic PathFinder: symbolic execution of Java bytecode}.
\newblock In Charles Pecheur, Jamie Andrews, and Elisabetta~Di Nitto, editors,
  {\em ASE}, pages 179--180. ACM, 2010.

\bibitem{SMA05}
K.~Sen, D.~Marinov, and G.~Agha.
\newblock {CUTE: a concolic unit testing engine for C}.
\newblock In {\em Proc.\ of ESEC/SIGSOFT FSE 2005}, pages 263--272. ACM, 2005.

\bibitem{SESGF11}
T.~Str{\"o}der, F.~Emmes, P.~Schneider-Kamp, J.~Giesl, and C.~Fuhs.
\newblock {A Linear Operational Semantics for Termination and Complexity
  Analysis of ISO Prolog}.
\newblock In {\em LOPSTR'11}, pages 237--252. Springer LNCS 7225, 2011.

\end{thebibliography}

\end{document}